\documentclass[11pt,a4paper]{article}
\usepackage{amsfonts}
\usepackage{amssymb}
\usepackage{amsmath}
\usepackage{cite,enumerate,amsthm,enumerate}

\usepackage{lscape}
\usepackage{diagbox}
\usepackage{multirow}
\usepackage{bigstrut}
\usepackage{float}

\newtheorem{theorem}{Theorem}[section]
\newtheorem{lemma}[theorem]{Lemma}
\newtheorem{proposition}[theorem]{Proposition}
\newtheorem{corollary}[theorem]{Corollary}

\newtheorem{example}[theorem]{Example}

\newcommand{\I}{{{I\hspace*{-0.2ex}}}}
\newcommand{\II}{{{I\hspace*{-0.5ex}I\hspace*{-0.2ex}}}}

\newcommand{\Ip}{{{I\hspace*{-0.2ex}}^\prime}}
\newcommand{\IIp}{{{I\hspace*{-0.5ex}I\hspace*{-0.2ex}}^\prime}}

\begin{document}
\markboth{Arunwan Boripan, Somphong Jitman, and Patanee Udomkavanich}
{Characterization and Enumeration of Complementary Dual Abelian Codes}

\title{Characterization and Enumeration of Complementary Dual Abelian Codes\thanks{This research was supported by the Thailand Research Fund under Research
    Grant MRG6080012.}
}


\author{Arunwan Boripan\thanks{A. Boripan is with the 
        Department of Mathematics and Computer Science, Faculty of Science,  Chulalongkorn University,   Bangkok 10330,  Thailand. 
        \texttt{boripan-arunwan@hotmail.com}} , Somphong Jitman\thanks{S. Jitman (Corresponding Author) is with the 
    Department of Mathematics, Faculty of Science,  Silpakorn University,  Nakhon Pathom 73000,  Thailand.    \texttt{sjitman@gmail.com}}, and Patanee Udomkavanich\thanks{P. Udomkavanich is with the 
Department of Mathematics and Computer Science, Faculty of Science,  Chulalongkorn University,   Bangkok 10330,  Thailand. \texttt{pattanee.u@chula.ac.th}}
}

\maketitle

\begin{abstract}
 Abelian codes and  complementary dual codes form   important classes of linear codes that have been extensively studied due to their rich algebraic structures and wide applications. In this paper, a family of abelian codes with complementary dual   in a group algebra $\mathbb{F}_{p^\nu}[G]$ has been studied under both the Euclidean and Hermitian inner products, where $p$ is a prime, $\nu$ is a positive integer, and $G$ is an arbitrary finite abelian group.
Based   on  the discrete Fourier transform decomposition for  semi-simple  group algebras and properties of ideas in local group algebras, the characterization of such codes have been  given. Subsequently,   the number of complementary dual abelian codes in  $\mathbb{F}_{p^\nu}[G]$  has been shown to be independent of the Sylow $p$-subgroup of $G$ and it has been completely determined for every  finite abelian group $G$.  In some cases, a simplified formula for the enumeration has been provided as well. 
The known results for cyclic complementary dual codes      can be viewed as corollaries. 
 
\end{abstract}

\noindent{Keywords: Complementary dual codes,  Abelian codes,  Group algebras, Local group algebra}

\noindent{Mathematics Subject Classification: 94B15,  94B05,   16A26}

 \section{Introduction} 

Abelian codes  form an important class of linear codes that has been extensively studied for both theoretical and practical reasons (see, for example,  \cite{BS2011}, \cite{Ch1992}, \cite{DKL2000}, \cite{JLLX2012},  \cite{JLS2014}, \cite{RS1992},  and \cite{S1993}).  Moreover, this  class of codes  contains the classical  cyclic codes which  can be efficiently encoded and decoded using shift registers.   Linear complementary dual (LCD) codes introduced in \cite{M1992}  constitute another interesting class of linear codes.  It  has been shown that   LCD codes are  asymptotically good and  they can be suitably applied    for the two-user  binary adder channel in  \cite{M1992}.  In \cite{Sen2004}, it has been shown that LCD codes meet the Gilbert-Varshamov bound.  LCD codes are also useful  in information protection from side-channel attacks and hardware Trojan horses (see \cite{CG2015},\cite{Carletetal2015},\cite{EHH2011}, \cite{ISW2003}, \cite{Ngoetal2014}, and  \cite{Ngoetal2015}). In \cite{GJG2016},  LCD codes have been applied in constructing good entanglement-assisted quantum error correcting codes. 
It is therefore of natural interest to study algebraic structured codes with complementary dual. Cyclic codes with complementary dual have been studied in \cite{YM1994}.  Necessary and sufficient conditions for cyclic codes to be complementary dual with respect to the Euclidean inner product  have been given in  \cite{YM1994}. 

In this paper,  we focus on a more general setup. Precisely,  abelian codes with complementary dual are studied  under both the Euclidean and Hermitian inner products.  We focus on  the characterization and enumeration of  such codes.  Based on the discrete Fourier transform decomposition for semi-simple group algebras, the group algebra $\mathbb{F}_{p^\nu}[G]$ can be view as a product of local group algebras for all finite abelian groups $G$. The complementary dual abelian codes and the direct summand ideals  in each  local component  are characterized. The  characterization and enumeration of complementary dual abelian codes can be obtained through the    discrete Fourier transform decomposition  and the characterization of ideals in local group algebras mentioned above.   A simplified formula for the enumeration is given as well  in  the case where the underlying group is a cyclic group  or a $q$-group, where  $q$ is a prime such that $q\neq p$. 

The paper is organized as follows.  Basic properties of group algebras and abelian codes are recalled in Section 2.  Complementary dual abelian codes in some local group algebras are studied in Section 3.  Based on the results in Section 3, it is shown that  the characterization  and enumeration of  Euclidean   complementary dual abelian codes over  finite fields of characteristics $p$  are  independent of the Sylow $p$-subgroup of $G$. The complete characterization and enumeration of Euclidean complementary dual abelian codes are given  in Section 4.  The analogous results for Hermitian complementary dual abelian codes are given in Section 5.  

\section{Preliminaries}

Let $R$ be a  commutative ring with identity and let $G$ be a finite abelian group, written additively.   Denote by    $R[G]$   the {\it group ring} of
$G$ over~$R$. The elements in $R[G]$ will be written as $\sum_{g\in G}\alpha_{{g }}Y^g $,
where $ \alpha_{g }\in R$.  The addition and the multiplication in $ R[G]$ are  given as in the usual polynomial rings over  $R$ with the indeterminate $Y$, where the indices are computed additively in $G$.    The homomorphism  $\epsilon:  R[G]\to R$ defined by \begin{align}
\label{aug}
\sum_{g\in G}\alpha_{{g }}Y^g  \mapsto \sum_{g\in G}\alpha_{{g }}
\end{align} is called  an {\em augmentation map} and the kernel $\Delta_R(G):=\ker (\epsilon)$ is called an {\em augmentation ideal} of $ R[G]$.

For  a prime $p$ and a positive integer $\nu$,    let  $\mathbb{F}_{p^\nu}$ denote   the  finite field of order  $p^\nu$.  In the case where $R =\mathbb{F}_{p^\nu}$, the group ring $\mathbb{F}_{p^\nu}[G]$ is called a {\em group algebra}.
An {\em abelian code} in  $\mathbb{F}_{p^\nu}[G]$  is   defined to be  an ideal in $\mathbb{F}_{p^\nu}[G]$ (see \cite{JLLX2012} and \cite{JLS2014}). The {\em Euclidean inner product}  between $u=\sum_{g\in G}u_{{g }}Y^g $ and  $v=\sum_{g\in G}v_{{g }}Y^g $   in $\mathbb{F}_{p^\nu}[G]$  is defined to be 
$\langle u,v\rangle _{\rm E}:= \sum_{g\in G} u_gv_g.$ The {\em Euclidean dual} of an abelian code $C$ in  $\mathbb{F}_{p^\nu}[G]$ is defined to be $C^{\perp_{\rm E}}:=\{v\in \mathbb{F}_{p^\nu}[G]\mid  \langle c,v\rangle _{\rm E} =0 \text{ for all } c\in C\}$.  An abelian code  $C$ in $\mathbb{F}_{p^\nu}[G]$  is said to be  {\em Euclidean complementary  dual}   if $ C\cap C^{\perp_{\rm E}}=\{0\}$.
In $\mathbb{F}_{{p^{2\nu}}}[G]$, the {\em Hermitian inner product}  between $u=\sum_{g\in G}u_{{g }}Y^g $ and  $v=\sum_{g\in G}v_{{g }}Y^g $   $\mathbb{F}_{p^{2\nu}}[G]$  is defined to be 
$\langle u,v\rangle _{\rm H}:= \sum_{g\in G} u_gv_g^{p^\nu}.$   
The {\em Hermitian dual} of an abelian code $C$ in  $\mathbb{F}_{p^{2\nu}}[G]$ is defined to be $C^{\perp_{\rm H}}:=\{v\in \mathbb{F}_{p^{2\nu}}[G]\mid  \langle c,v\rangle _{\rm H} =0 \text{ for all } c\in C\}$. 
An abelian code  $C$ in $\mathbb{F}_{p^{2\nu}}[G]$  is said to be  {\em Hermitian complementary  dual}    if $ C\cap C^{\perp_{\rm H}}=\{0\}$.

For  positive integers $i$ and $j$ with $\gcd(i,j)=1$, let ${\rm ord}_j(i)$ denote the multiplicative order of $i$ modulo $j$.   
Let $P$ denote the Sylow $p$-subgroup of $G$. Then $G\cong  A\times P$, where $A$ is a subgroup $G$ such that $|A|=[G:P]$ and $p\nmid |A|$.   
For each $a\in A$, denote by ${\rm ord}(a)$ the additive order of $a$ in $A$. A {\it $p^\nu$-cyclotomic class}   of $A$ containing $a\in A$, denoted by $S_{p^\nu}(a)$, is defined to be the set
\begin{align*}
S_{p^\nu}(a):=&\{p^{ \nu i}\cdot a \mid i=0,1,\dots\}
=\{p^{ \nu i}\cdot a \mid 0\leq i< {\rm ord}_{{\rm ord}(a)}(p^\nu) \}, 
\end{align*}
where $p^{\nu i}\cdot a:= \sum\limits_{j=1}^{p^{\nu i}}a$ in $A$.

First, we consider the decomposition of    $\mathcal{R}:=\mathbb{F}_{p^{\nu}}[A]$.  In this case, $\mathcal{R}$ is semi-simple (see \cite{Be1967_2}) which can be  decomposed  using the Discrete Fourier Transform  in  \cite{RS1992}   (see \cite{JLS2014} and \cite{JLLX2012} for more details).  For completeness, the decomposition used in this paper is  summarized as follows.

An {\em idempotent} in $\mathcal{R}$ is a nonzero element $e$ such that $e^2=e$. It is called {\em primitive} if for every other idempotent $f$, either $ef=e$ or $ef=0$.  The primitive idempotents in $\mathcal{R}$ are induced by the $p^\nu$-cyclotomic classes of $A$ (see \cite[Proposition II.4]{DKL2000}).   Let $\{ a_1,a_2,\dots, a_t\}$ be a complete set of representatives of $p^\nu$-cyclotomic classes  of $A$  and let $e_i$ be the primitive idempotent induced by $S_{p^{\nu}}(a_i)$  for all $1\leq i\leq t$.  From \cite{RS1992},   $\mathcal{R}$ can be decomposed as 
\begin{align}\label{eq-decom0}
\mathcal{R}= \bigoplus_{i=1}^t\mathcal{R}e_i .
\end{align}
It is well known (see \cite{JLLX2012}  and \cite{JLS2014}) that $\mathcal{R}e_i\cong \mathbb{F}_{p^{\nu s_i}}$,  where  $s_i=|S_{p^{\nu}}(a_i)|$  provided that   $e_i$ is induced by $S_{p^{\nu}}(a_i)$, and hence,
\begin{align}\label{eq-decom01}
\mathbb{F}_{p^{\nu}}[A\times P]  
\cong \bigoplus_{i=1}^t(\mathcal{R}e_i )[P]\cong \prod_{i=1}^t \mathbb{F}_{p^{\nu s_i}}[P]   .
\end{align}
Therefore, every abelian code $C$ in  $  \mathbb{F}_{p^{\nu}}[A\times P]  $ can be viewed as   

\begin{align} \label{cyclicrep} C\cong \prod_{i=1}^t C_i, \end{align}
where $C_i$ is  an abelian code  in  $\mathbb{F}_{p^{\nu s_i}}[P]$ for all $i=1,2,\dots, t$.  

The goal of this paper is to characterize and enumerate complementary dual abelian codes in $\mathbb{F}_{p^{\nu }}[G]$ for all finite abelian groups $G$, primes $p$, and positive integers $\nu$ .  The characterization of each component  in the decomposition  \eqref{eq-decom01} is given in Section 3. Combining the results in Sections 2 and 3, the characterization and enumeration of  complementary dual abelian codes in $\mathbb{F}_{p^{\nu }}[G]$ are given in Sections 4-5.

\section{Abelian Codes in some Local Group Algebras}
In this section,     we focus on   complementary dual abelian codes and direct summand ideals    in  each component  $\mathbb{F}_{p^\nu}[P]$ in the decomposition \eqref{eq-decom01}. 

For  a  finite commutative ring $R$ with identity, the {\em Jacobson radical}  of    $R$, denoted by  $Jac(R)$,   is defined to be the intersection of  all  maximal ideals of $R$.  The ring $R$ is said to be
\textit{local} if it has a unique maximal ideal and it  is called a
\textit{chain ring} if its ideals are  linearly ordered by
inclusion.  A finite commutative chain ring $R$  with maximal ideal $M$ is said to have {\em nilpotency index} $e$ if   $e$ is the smallest positive integer such that $M^e=0$.   The following facts are well known. 

\begin{proposition}\label{ring1}
    A finite commutative  ring $R$ with identity is a local ring if
    and only if the set of all non-invertible elements in $R$ forms a
    maximal ideal.
\end{proposition}

\begin{proposition}\label{ring2}
    A finite commutative  local ring $R$ is a chain ring if and only if its
    maximal ideal is principal.
\end{proposition}
A local group ring has been  characterized in the following lemma.
\begin{lemma}[{\cite[Theorem]{N1972}}]  \label{lemLocal}Let $R$ be a commutative ring with identity and let $G$ be a finite  abelian group. Then  $R[G]$ is   local if and only if $R$ is local, $G$ is a $p$-group and
    $p\in Jac(R)$, where $Jac(R)$ is the Jacobson radical of $R$. 
\end{lemma}

From Lemma \ref{lemLocal},  a group algebra  $\mathbb{F}_{p^\nu}[P]$ is  local for all $p$-groups $P$.   More properties of  a  local group algebra  $\mathbb{F}_{p^\nu}[P]$  can be deduced as follows. 

\begin{proposition} \label{propLocal}
    Let $p$ be a prime and let $\nu$ be a positive  integer. Let $P$ be a finite abelian $p$-group.  Then the following statements holds.
    \begin{enumerate}
        \item   $\mathbb{F}_{p^\nu}[P]$ is a local group algebra with maximal ideal $ \Delta_{\mathbb{F}_{p^\nu}}(P)$. 
        \item    $\mathbb{F}_{p^\nu}[\mathbb{Z}_{p^k}]$ is a finite chain ring  with nilpotency index $p^k$ for all positive integers $k$. .
    \end{enumerate}
\end{proposition}
\begin{proof}
    Since $\mathbb{F}_{p^\nu}$  is a field of characteristic $p$,  $\mathbb{F}_{p^\nu}$ is local and $p\in \{0\}= Jac(\mathbb{F}_{p^\nu})$. By Lemma \ref{lemLocal}, $\mathbb{F}_{p^\nu}[P]$ is local.   Using the map $\epsilon$ defined in (\ref{aug}) and the first isomorphism for rings, we have  \[ \mathbb{F}_{p^\nu}\cong \mathbb{F}_{p^\nu}[P] /\ker (\epsilon) \cong   \mathbb{F}_{p^\nu}[P]  /  \Delta_{\mathbb{F}_{p^\nu} }(P),\]
    and hence,   $\Delta_{\mathbb{F}_{p^\nu} }(P)$ is the maximal ideal of $\mathbb{F}_{p^\nu}[P]$. This completes the proof of  $1$.

    It is not difficult to see that  the maximal ideal  $\Delta_{\mathbb{F}_{p^\nu} }(\mathbb{Z}_{p^k})\cong \langle Y-1\rangle$  is principal. By Proposition \ref{ring2}, $\mathbb{F}_{p^\nu}[\mathbb{Z}_{p^k}]$ is a finite chain ring. Since $(Y-1)^{p^k}=0$ and  $(Y-1)^{p^k-1}\ne 0$, $\mathbb{F}_{p^\nu}[\mathbb{Z}_{p^k}]$ has nilpotency index $p^k$. This completes the proof of $2$.  \end{proof}

\begin{corollary}  Let $p$ be a prime and let $\nu$ be a positive integer.  Then   $\mathbb{F}_{p^\nu}[\mathbb{Z}_{p^k}]$ contains $p^k+1$  ideals. 
\end{corollary}
\begin{proof}
    The statement follows from Proposition \ref{propLocal}.  Precisely, the ideals in $\mathbb{F}_{p^\nu}[\mathbb{Z}_{p^k}]$ are of the form $\langle (Y-1)^i\rangle $ for all $0\leq i\leq p^k$.
\end{proof}

The characterizations of the Euclidean and Hermitian  complementary dual abelian codes and  the direct summands in a local group algebra $\mathbb{F}_{p^\nu}[P]$  are  given in the following theorems. 

\begin{theorem} \label{complementary}
    Let $p$ be a prime and let $\nu\geq 1$ be an integer.  Let $P$ be a finite abelian $p$-group.  Then $\{0\}$ and $\mathbb{F}_{p^\nu}[P]$ are the only Euclidean complementary  dual abelian codes in  $\mathbb{F}_{p^\nu}[P]$.
\end{theorem}
\begin{proof}  Clearly, $\{0\}$ and $\mathbb{F}_{p^\nu}[P]$ are Euclidean  complementary  dual abelian codes in  $\mathbb{F}_{p^\nu}[P]$. 
    Assume that there exists a Euclidean  abelian code $C$ in  $ \mathbb{F}_{p^\nu}[P]$ such that  $\{0\}\subsetneq C \subsetneq \mathbb{F}_{p^\nu}[P]$. By  Proposition \ref{propLocal}, $\mathbb{F}_{p^\nu}[P]$ is   local  with maximal ideal $\Delta_{\mathbb{F}_{p^\nu} }(P)$.  Then $C\subseteq \Delta_{\mathbb{F}_{p^\nu} }(P)$. It follows that  $\Delta_{\mathbb{F}_{p^\nu} }(P) ^{\perp_{\rm E}} \subseteq C^{\perp_{\rm E}}  \subseteq \Delta_{\mathbb{F}_{p^\nu} }(P)$ which implies  $\Delta_{\mathbb{F}_{p^\nu} }(P)^{\perp_{\rm E}}  \subseteq C \subseteq \Delta_{\mathbb{F}_{p^\nu} }(P)$. Hence,  $\{0\}\ne \Delta_{\mathbb{F}_{p^\nu} }(P)^{\perp_{\rm E}}  \subseteq C\cap C^{\perp_{\rm E}} \subseteq \Delta_{\mathbb{F}_{p^\nu} }(P)$. Consequently, $C$ is not  Euclidean complementary dual. Therefore, the ideals $\{0\}$ and $\mathbb{F}_{p^\nu}[P]$ are the only Euclidean complementary  dual abelian codes in  $\mathbb{F}_{p^\nu}[P]$.
\end{proof}

It is not difficult to see that the proof of Theorem \ref{complementary} is independent of the inner product. Hence, we have the following corollary. 

\begin{corollary} \label{Hcomplementary}
    Let $p$ be a prime and let $\nu$ be a positive integer.  Let $P$ be a finite abelian $p$-group.  Then $\{0\}$ and $\mathbb{F}_{p^{2\nu}}[P]$ are the only Hermitian complementary  dual abelian codes in  $\mathbb{F}_{p^{2\nu}}[P]$.
\end{corollary}

\begin{theorem}\label{derectsummand}
    Let $p$ be a prime and let $\nu$ be a positive integer.  Let $P$ be a finite abelian $p$-group.  Then ideals $\{0\}$ and $\mathbb{F}_{p^\nu}[P]$ are the only direct summands in  $\mathbb{F}_{p^\nu}[P]$.
\end{theorem}
\begin{proof}
    Let $\{0\}\subsetneq C \subsetneq \mathbb{F}_{p^\nu}[B]$  be  an  ideal in  $\mathbb{F}_{p^\nu}[P]$. Suppose that $C$ is  a direct summand. Then there exist an ideal $C^\prime $ in $\mathbb{F}_{p^\nu}[P]$ such that $C\cap C^\prime=\{0\}$ and $C+C^\prime =\mathbb{F}_{p^\nu}[P]$. Since $\Delta_{\mathbb{F}_{p^\nu} }(P)$ is the maximal ideal in $\mathbb{F}_{p^\nu}[P]$, we have  $C^\prime \subseteq  \Delta_{\mathbb{F}_{p^\nu} }(P) $.  Hence, $C+C^\prime\subseteq \Delta_{\mathbb{F}_{p^\nu} }(P) \subsetneq \mathbb{F}_{p^\nu}[P]$, a contradiction. Therefore,  the ideals $\{0\}$ and $\mathbb{F}_{p^\nu}[P]$ are the only direct summands in  $\mathbb{F}_{p^\nu}[P]$.
\end{proof}

The next corollary follows immediately. 
\begin{corollary}
    Let $p$ be a prime and let $\nu$ be a positive integer.  Let $P$ be a finite abelian $p$-group.  Then the following statements hold.
    \begin{enumerate}
        \item    The number of Euclidean complementary dual abelian codes in $\mathbb{F}_{p^\nu}[P]$   is $2$.
        \item    The number of  direct summand ideals in $\mathbb{F}_{p^\nu}[P]$ is $2$. 
        \item    The number of Hermitian  complementary dual abelian codes in $\mathbb{F}_{p^{2\nu}}[P]$   is $2$.
        \item    The number of   direct summand ideals in $\mathbb{F}_{p^{2\nu}}[P]$ is $2$. 
    \end{enumerate}
\end{corollary}

\section{Euclidean  Complementary Dual Abelian Codes}

In this section,  we focus on  characterization and enumeration of Euclidean  complementary dual abelian codes in  $\mathbb{F}_{p^\nu}[G]$, where  $p $ is a prime, $\nu $ is a positive integer, and $G$ is an arbitrary finite abelian group.

\subsection{Characterization}

For a finite abelian group $G$, write $G=A\times P$, where $p\nmid |A|$ and $P$ is the Sylow $p$-subgroup of $G$. Based on the results in Section 3,   the characterization and enumeration of  complementary dual abelian codes in   $\mathbb{F}_{p^\nu}[G]$ are given and they are shown to be   independent of $P$.  

The decomposition of  $\mathbb{F}_{p^\nu}[A]$ in  \cite{JLLX2012} plays a vital role in the study of Euclidean complementary dual abelian codes in $\mathbb{F}_{p^\nu}[G]$ and it is recalled as follows. 
A $p^\nu$-cyclotomic class $S_{p^{\nu}}(a)$ is said to be of  {\em type} ${\I}$     if $S_{p^{\nu}}(a)=S_{p^{\nu}}(-a)$, or {\em type} ${\II}$   if $S_{p^{\nu}}(-a)\neq S_{p^{\nu}}(a)$. 	 
Without loss of generality, the representatives $a_1, a_2, \dots, a_t$ of  $p^\nu$-cyclotomic classes   of $A$  can be  chosen such that $\{a_j| j=1,2,\dots,{r_{\I}}\}$ and $\{a_{r_{\I}+l}, a_{r_{I}+r_{\II}+l}=-a_{r_{I}+l} \mid l=1,2,\dots, r_{\II}\}$ are  sets of representatives of $p^\nu$-cyclotomic classes of $A$ of types $\I$ and ${\II}$, respectively, where $t=r_{\I}+2r_{\II}$.    Assume the notations used  in Section 2, we have   $s_i=|S_{p^\nu}(a_i)|$ for all $1\leq i\leq t$  and  $s_{r_\I+l}=s_{r_\I+r_\II+l}$ for all $1\leq l\leq r_\II$.

Rearranging the terms in the decomposition of  $\mathcal{R}$  in  \eqref{eq-decom01}  based on these  $2$ types of cyclotomic classes (see \cite{JLLX2012}),  we have 
\begin{align}
\mathbb{F}_{p^{\nu}}[A\times P]    & 
\cong   \left( \prod_{j=1}^{r_{\I}} \mathbb{K}_j[P]  \right) \times \left( \prod_{l=1}^{r_{\II}} (\mathbb{L}_l[P]\times \mathbb{L}_l[P] )  \right), \label{eqSemiSim}
\end{align}
where   $ \mathbb{K}_j\cong \mathbb{F}_{p^{\nu s_{j}}}  $ for all   $j=1,2,\dots, r_{\I}$  and  $ \mathbb{L}_l  \cong \mathbb{F}_{  p^{\nu s_{r_\I+l}}}   $      for all $l=1,2,\dots, r_{\II}$.

From \eqref{eqSemiSim},  it follows that an abelian code $C$ in  $\mathbb{F}_{p^{\nu}}[A\times P]   $ can be viewed as 
\begin{align}\label{decomC} 
C\cong   \left(\prod_{j=1}^{r_{\I}} C_j  \right)\times \left(\prod_{l=1}^{r_{\II}} \left( D_l\times D_l^\prime\right) \right), \end{align}
where $C_j$, $D_s$ and $D_s^\prime$ are   abelian  codes in        $\mathbb{K}_j[P]$, $\mathbb{L}_l[P]$ and $\mathbb{L}_l[P]$, respectively,  for all    $j=1,2,\dots,r_{\I}$ and  $l=1,2,\dots,r_{\II}$.

From   \cite[Section II.D]{JLLX2012},  the Euclidean dual of $C$  in (\ref{decomC}) is of the 
form 

\begin{align} \label{eq-Edual}
C^{\perp_{\rm E}}\cong    \left(\prod_{j=1}^{r_{\I}} C_j ^{\perp_{\rm H}} \right)\times \left(\prod_{l=1}^{r_{\II}} \left( (D_l^\prime) ^{\perp_{\rm E}}\times  D_l^{\perp_{\rm E}}\right) \right), 
\end{align}
where $\perp_{\rm H}:=\perp_{\rm E}$ if $\mathbb{K}_j\cong  \mathbb{F}_{p^k}$.

The characterization of  a Euclidean complementary  dual abelian code in  $\mathbb{F}_{p^{\nu}}[G]   $  is given in the following proposition. 

\begin{proposition} \label{charLCDE} Let $p$ be a prime and let $\nu$  be a positive  integer. Let  $A$ be  finite abelian group such that $p\nmid |A|$ and let  $P$ be a finite abelian $p$-group. Then an abelian code $C$ in $\mathbb{F}_{p^{\nu}}[A\times P]  $ decomposed as in  \eqref{decomC} is Euclidean complementary dual if and only if the following statements hold.
    \begin{enumerate}
        \item  $C_j$ is Hermitian  complementary dual   for all $1\leq j\leq r_\I$.
        \item  $D_l\cap (D_l^\prime)^{\perp_{\rm E}}=\{0\}$ and    $D_l^\prime \cap D_l^{\perp_{\rm E}}=\{0\}$  for all $1\leq l\leq r_{\II}$.
    \end{enumerate}
\end{proposition}
\begin{proof}
    The result can be deduced directly from \eqref{decomC} and  \eqref{eq-Edual}.
\end{proof}

\begin{corollary} \label{corCharLCDE}
    Let $p$ be a prime and let $\nu$  be a positive  integer. Let  $A$ be  finite abelian group such that $p\nmid |A|$ and let  $P$ be a finite abelian $p$-group. Then an abelian code $C$ in $\mathbb{F}_{p^{\nu}}[A\times P]  $ decomposed as in  \eqref{decomC} is Euclidean complementary dual if and only if the following statements hold.
    \begin{enumerate}
        \item  $C_j\in \{\{0\}, \mathbb{K}_j[P]\}$      for all $1\leq j\leq r_\I$.
        \item  $(D_l, D_l^\prime ) \in \{ (\{0\}, \mathbb{L}_j[P]), (\mathbb{L}_j[P],\{0\})\}$  for all $1\leq l\leq r_{\II}$.
    \end{enumerate}
\end{corollary}
\begin{proof}
    From Proposition \ref{charLCDE},   $C_j$ is Hermitian  complementary dual   for all $1\leq j\leq r_\I$ and 
    $D_l$ is a direct summand  for all $1\leq l\leq r_{\II}$. By Corollary \ref{Hcomplementary},  $C_j\in \{\{0\}, \mathbb{K}_j[P]\}$      for all $1\leq j\leq r_\I$.  By Theorem \ref{derectsummand}, $D_l \in \{ \{0\}, \mathbb{L}_l[P] \}$  for all $1\leq l\leq r_{\II}$, and hence, the result follows. 
\end{proof}

From Corollary \ref{corCharLCDE}, it is not difficult to see that  the number of Euclidean complementary dual abelian codes in $\mathbb{F}_{p^{\nu}}[A\times P]  $ is  independent of  $P$ and it is determined in the following corollary. 
\begin{corollary}\label{corNumLCDE}
    Let $p$ be a prime and let $\nu$  be  a positive integer. Let  $A$ be  finite abelian group such that $p\nmid |A|$ and let  $P$ be a finite abelian $p$-group.  If  $\mathbb{F}_{p^{\nu}}[A\times P]  $ decomposed as in  \eqref{eqSemiSim}, then the number of Euclidean complementary dual abelian codes in $\mathbb{F}_{p^{\nu}}[A\times P]  $ is  \[2^{r_{\I}+r_{\II}}\]
\end{corollary}
\begin{proof}
    From the characterization in Corollary \ref{corCharLCDE},  the number of choices of      $C_j$   is    $2^{r_\I}$ and  the number of choices of 
    $(D_l, D_l^\prime )  $  is  $2^{r_{\II}}$. Hence, the  number of Euclidean complementary dual abelian codes in $\mathbb{F}_{p^{\nu}}[A\times P]  $ is $2^{r_{\I}+r_{\II}}$ as desired. 
\end{proof}

\subsection{Enumerations}
As discussed in Corollary \ref{corNumLCDE},  the number of Euclidean complementary dual abelian codes in $\mathbb{F}_{p^{\nu}}[A\times  P]  $ is  independent of $P$. Here, it is  sufficient  to  determined the number  ${r_{\I}+r_{\II}}$ from the group algebra  $\mathbb{F}_{p^{\nu}}[A]$.

Let $\chi$   be a function defined by 
\begin{align} \label{r12}
\chi(d,p^\nu)=
\begin{cases}
1 &\text{if there exists a positive integer } i \text{ such that }  d|(p^{\nu i}+1),\\
0 &\text{otherwise} .\\		
\end{cases}
\end{align}

The  following lemma is a straightforward generalization  of  the case where $p=2$  (see {\cite[Lemma 4.5]{JLLX2012}}). 

\begin{lemma}\label{propType} Let $p$ be  a prime and let $\nu$ be a positive integer.  Let $A$ be a finite  abelian group    such that $p\nmid |A|$  and let  $a\in A$. Then  $S_{p^\nu}(a) $ is of type $\I$  if and only if $\chi({\rm ord}(a) ,p^\nu)=1$. 
\end{lemma}

The value ${r_{\I}+r_{\II}}$ of the group algebra $\mathbb{F}_{p^{\nu}}[A]$  is determined in the following proposition. 
\begin{proposition} \label{propEnum}
    Let $p$ be a prime and let $\nu$ be a positive integer. Let $A$ be  a finite abelian group of exponent $N$. If $p\nmid N$, then 
    \[{r_{\I}+r_{\II}}=\sum_{d|N}\chi(d,p^\nu)\frac{\mathcal{N}_A(d) }{ {\rm ord}_d(p^\nu)} +\frac{1}{2} \sum_{d|N}\left(1-\chi(d,p^\nu)\right)\frac{\mathcal{N}_A(d) }{ {\rm ord}_d(p^\nu)} ,\]
    where  $\mathcal{N}_A(d)$ denotes the number of elements of order $d$ in $A$ determined in \cite{B1997}.
\end{proposition}
\begin{proof}
    Using the arguments similar to those in \cite[Remark 2.5]{JLS2014},  for each $d|N$, the elements of order $d$ in $A$  are partitioned into $p^{\nu}$-cyclotomic classes of the same size $ \frac{\mathcal{N}_A(d) }{ {\rm ord}_d(p^\nu)} $.  By  Lemma \ref{propType}, it follows that  \[{r_{\I}}=\sum_{d|N}\chi(d,p^\nu)\frac{\mathcal{N}_A(d) }{ {\rm ord}_d(p^\nu)} \text{  and } r_{\II} =\frac{1}{2} \sum_{d|N}\left(1-\chi(d,p^\nu)\right)\frac{\mathcal{N}_A(d) }{ {\rm ord}_d(p^\nu)} .\]
    Hence, 
    the result follows.
\end{proof}

In the case where $A$ is a cyclic group of order $n$ with $p\nmid n$,  it follows  that the exponent of $A$ is $n$ and $\mathcal{N}_A(d)=\Phi(d)$ for all divisors $d$ of $n$, where  $\Phi(d)$ denotes the  Euler totient phi function. Hence,  the following corollary can be deduced. 
\begin{corollary}
    Let $p$ be a prime and let $\nu$ be a positive integer. Let $n$ be a positive integer such that  $p\nmid n$. Then  the number of Euclidean complementary dual cyclic codes of length $np^k$ over $\mathbb{F}_{p^\nu}$ equals  the number of Euclidean complementary dual cyclic codes of length $n$ over $\mathbb{F}_{p^\nu}$ for all   integers $k\geq 0$ which is 
    \[2^{\sum_{d|n}\chi(d,p^\nu)\frac{\Phi(d) }{ {\rm ord}_d(p^\nu)} +\frac{1}{2} \sum_{d|n}\left(1-\chi(d,p^\nu)\right)\frac{\Phi(d) }{ {\rm ord}_d(p^\nu)} }.\]
\end{corollary}

From Proposition \ref{propEnum},  for an arbitrary finite abelian group $A$ with $p\nmid |A|$,    we have ${r_{\I}+r_{\II}}\geq 1$.   Hence, there  exist at least two Euclidean complementary dual abelian codes  in  $\mathbb{F}_{p^\nu}[A]$.

    In Table \ref{T1}, the  number  ${r_{\I}+r_{\II}}$  for   group algebras $\mathbb{F}_2[A]$ is given for abelian groups $A$ of odd order less than $50$.
    
    \begin{table}[!hbt]
        \centering
    \begin{tabular}
       {|r|r|r|} \hline Order of $A$ & $A$~~~~~~ & $r_\I+r_\II$~~\\
       \hline  
       $3$ &$\mathbb{Z}_3$ &$2 $\\
       $5$ &$\mathbb{Z}_5 $ &$ 2$\\
       $7$ &$\mathbb{Z}_7 $ &$ 2$\\
       $9$ &$\mathbb{Z}_3\times\mathbb{Z}_3 $ &$ 5$\\
       &$\mathbb{Z}_{3^2} $ &$ 3$\\
       $11$ &$\mathbb{Z}_{11}$ &$2 $\\
       $13$ &$\mathbb{Z}_{13} $ &$2 $\\
       $15$ &$\mathbb{Z}_5\times\mathbb{Z}_3 $ &$ 4$\\
       $17$ &$\mathbb{Z}_{17} $ & $3 $\\
       $19$ &$\mathbb{Z}_{19} $ &$2 $\\
       $21$ &$\mathbb{Z}_7\times\mathbb{Z}_3 $ &$4 $\\
       $23$ &$\mathbb{Z}_{23} $ &$2 $\\
       $25$ &$ \mathbb{Z}_5\times\mathbb{Z}_5 $ &$ 7$\\
       &$\mathbb{Z}_{5^2} $ &$ 3$\\
       $27$ &$\mathbb{Z}_3\times\mathbb{Z}_3\times\mathbb{Z}_3 $ &$14 $\\
       &$\mathbb{Z}_{3^2}\times\mathbb{Z}_3 $ &$8$\\
       &$\mathbb{Z}_{3^3}$ &$ 4$\\
       $29$ &$\mathbb{Z}_{29}$ &$2 $\\
       $31$ &$\mathbb{Z}_{31} $ &$4 $\\
       $33$ &$\mathbb{Z}_{11}\times\mathbb{Z}_3 $ &$ 5$\\
       $35$ &$\mathbb{Z}_{7}\times\mathbb{Z}_5 $ &$ 4$\\
       $37$ &$\mathbb{Z}_{37} $ &$2 $\\
       $39$ &$\mathbb{Z}_{13}\times\mathbb{Z}_3 $ &$4$\\
       $41$ &$\mathbb{Z}_{41} $ &$3 $\\
       $43$ &$\mathbb{Z}_{43}$ &$4 $\\
       $45$&$\mathbb{Z}_5\times\mathbb{Z}_3\times\mathbb{Z}_3$ &$ 10$\\
       &$\mathbb{Z}_5\times\mathbb{Z}_{3^2} $ &$6 $\\
       $47$ &$\mathbb{Z}_{47} $ &$3 $\\
       $49$ &$ \mathbb{Z}_7\times\mathbb{Z}_7$ &$ 9$\\
       &$ \mathbb{Z}_{7^2}$&$ 3$\\
       \hline 
   \end{tabular}
        \caption{The number  $r_\I+r_\II $ for  group algebra  $\mathbb{F}_2[A]$} \label{T1}
    \end{table}

In the following subsections,  a simplified formula  for  \eqref{r12} is  given for some families of  finite abelian $q$-groups, where $q$ is a prime such that $p\ne q$. 
Let $A\cong(\mathbb{Z}_{q^k})^s$, where $k$ and $s$ are  positive integers, and $q$ is  prime such that $\gcd(p,q)=1$. For each $0 \leq i \leq k$, define \[A_{q^i}:=\{a \in A |\: {\rm ord}(a)=q^i\}.\] 
Clearly,  $A_1, A_q, \ldots, A_{q^k}$ are pair-wise disjoint and $A = A_1 \cup A_q \cup \cdots \cup A_{q^k}$. For each $1 \leq i \leq k$, it is not difficult to see that $A_{q^i} = \left(q^{k-i}\mathbb{Z}_{q^k}\right)^s \setminus \left(q^{k-(i-1)}\mathbb{Z}_{q^k}\right)^s$.  Consequently, we have $|A_1|=1$ and
\begin{align*}\label{Eq:Cardinality}
|A_{q^i}| = q^{is}-q^{(i-1)s}
\end{align*}
for all  $i = 1, 2,  \ldots, k$.   Note that $|S_{p^\nu}(a)|= {\rm ord}_{{\rm ord}(a)}(p^\nu) = {\rm ord}_{q^i}(p^\nu) $ for all $a \in A_{q^i}$.

\subsubsection{$A=(\mathbb{Z}_{2^k})^s$ and $p$ is an odd prime}

    Here, we consider the case where $p$ is odd and $q=2$, i.e., $A\cong(\mathbb{Z}_{2^k})^s$.  Clearly,  $A_1=\{0\}$ and $S_{p^\nu}(0)$ is of type $\I$.  For an element   $a \in A_2$, we have  ${\rm ord}(a)=2$  which implies that  $a= - a$. Hence,  $S_{p^\nu}(a)=S_{p^\nu}( -a) $ is of type $\I$ and its cardinality  is  $1$.  
 In general, we have the following characterizations of  the $p^\nu$-cyclotomic classes of type $I$ in $A$.
    
    \begin{lemma} \label{lem2p}
   For each $i\in \{0,1,2,\dots, k\}$, let $a\in A_{2^i}$. Then  one of the following statements holds.
   \begin{enumerate}
 \item   If $i\in \{0,1\}$, then  $S_{p^\nu}(a)$  is of type $I$   and size $1$.  
 \item   
   If $i\geq 2$, then  $S_{p^\nu}(a)$  is of type $I$ if and only if ${\mathrm ord}_{2^i}(p^\nu)=2$.
\end{enumerate}   
    \end{lemma}
    \begin{proof}
  The first statement follows from  the discussion above. 
 To prove the second statement, assume that $i\geq 2$.  Assume that  $S_{p^\nu}(a)$  is of type $I$.  Then $-a\in  S_{p^\nu}(a)$.  
  Then  $p^{\nu j} \cdot a=-a$ for some integer $j\geq 1$.
  Since $i\geq 2$, we have  $4|(p^{\nu j}+1)$. If $j$ is even, then    $p^{\nu j}\equiv 1 \ {\rm mod\ } 4$  which implies that $(p^{\nu j}+1)\equiv 2\, {\rm mod }\, 4$, a contradiction. 
  It follows that   $j$ is odd.  Since $p^{\nu j}+1=(p^{\nu j}+1)\left(\sum\limits_{i=0}^{j-1} (-1)^ip^{\nu (j-1-i)} \right)$ and  $\sum\limits_{i=0}^{j-1} (-1)^ip^{\nu(j-1-i)} $ is odd, we have  $2^i|(p^\nu+1)$. Hence,  $p^\nu \equiv -1 \ {\rm mod\ } 2^i$  which implies that  ${\rm ord}_{2^i}( p^\nu)=2$.

  Conversely, assume that ${\mathrm ord}_{2^i}(p^\nu)=2$.  Then $p^{2\nu}\cdot a=a$ and $p^{\nu}\cdot a\ne a$   which implies that  $p^{\nu}\cdot a =- a$. Hence,  $S_{p^\nu}(a)$
   is a $p^\nu$-cyclotomic class of type $I$. 
    \end{proof}

    \begin{proposition}\label{Prop:t2} Let $0\leq r^\prime \leq k$ be the largest integer such that $2^{r^\prime}|(p^\nu+1)$ and let $a\in A$.  Then $S_{p^\nu}(a)$ is of type $I$ if and only if $a\in A_{2^i}$ for some $0\leq i\leq r^\prime$.  Equivalently,  $a\in A$ is of type $\II$ if and only if $a\in A_{2^i}$ for some $i> r^\prime$.
    \end{proposition}
    \begin{proof}
     Assume that  $i >  r^\prime$.  Then $2^i\nmid (p^\nu +1)$ which implies that ${\mathrm ord} _{2^i}(p^\nu) >2$. By Lemma \ref{lem2p},  $S_{p^\nu}(a)$ is not of type $I$.

     Conversely,  assume that  $a\in A_{2^i}$ for some $0\leq i\leq r^\prime$.  From Lemma \ref{lem2p}, if $i\in \{0,1\}$, then $S_{p^\nu}
(a)$ is of type $I$.     Assume that $2\leq i \leq r^\prime$. Then $2^i\geq 4$ and  $2^i| (p^\nu +1)$  which implies that   ${\mathrm ord} _{2^i}(p^\nu) =2$. By Lemma \ref{lem2p},  it follows that  $S_{p^\nu}(a)$ is of type $I$.
\end{proof}

    Using Proposition~\ref{Prop:t2},  $A_{2^i}$ contains $p^\nu$-cyclotomic classes of type $\I$ for all $i=0, 1, \ldots, r'$ and the rest of the sets $A_{2^j}$ contain $p^\nu$-cyclotomic classes of type $\II$, for $j=r'+1, \ldots, k$.  Since ${\mathrm ord}_{2^i}(p^\nu)=1$ for all $i\in \{0,1\}$ and  ${\mathrm ord}_{2^i}(p^\nu)=2$ for all  $2 \leq i \leq r'$,  we have \[r_I = \sum_{i=0}^{r'} \frac{|A_{2^{i}}|}{{\mathrm ord}_{2^i}(p^\nu)} =|A_1|+|A_2|+\sum_{i=2}^{r'} \frac{|A_{2^{i}}|}{2} = {2^s +\frac{2^{r's}-2^s}{2} }\] and 
    \[r_{II} = \sum_{i=r'+1}^k\frac{|H_{2^{i}}|}{2{\mathrm ord}_{2^i}(p^\nu)} = \sum_{i=r'+1}^k\frac{2^{is} - 2^{(i-1)s}}{2{\mathrm ord}_{2^i}(p^\nu)}.\]  
    
    Hence,   the number of Euclidean complementary dual abelian codes in $\mathbb{F}_{p^\nu}[( \mathbb{Z}_{2^k})^s]$ follows immediately from Corollary \ref{corNumLCDE}.

\subsubsection{$A=(\mathbb{Z}_{q^k})^s$, where $q$ is an odd prime and $\gcd(q,p)=1$}

In this part,   we focus on the case where $A=(\mathbb{Z}_{q^k})^s$, where $q$ is an odd prime and $\gcd(q,p)=1$. Note that $S_{p^\nu}(0) = \{0\} = A_1$ is of type $\I$.

Next, we characterize the type of   $p^\nu$-cyclotomic classes of elements in $A_q$ in terms  of ${\rm ord}_q(p^\nu)$.

\begin{lemma}
[{\cite[Proposition 4]{M1997}}]\label{ord} Let $q$ be an odd prime and let $i$ be a positive integer. If $p$ is a prime such that $q\ne p$, then ${\rm ord}_{q^i}(p)={\rm ord}_{q}(p)q^j$ for some $j\geq 0$. 
\end{lemma}

\begin{proposition}\label{Prop:TypeI:Zpks} 
    Let $a \in A_q$.  Then  $S_{p^\nu}(a)$ is of type $\I$ if and only if ${\rm ord}_q(p^\nu)$ is even.   Equivalently,   $S_{p^\nu}(a)$ is of type $\II$ if and only if  ${\rm ord}_{q}(p^\nu)$ is odd.
\end{proposition}

\begin{proof} 
    The first implication follows from   \cite[Remark 3.1]{JL2015}.  
    
    Conversely, assume that ${\rm ord}_q(p^\nu)=2j$ is   even, where $j$ is a positive integer.  
    Then $ a=p^{\nu {\rm ord}_q(p^\nu)}\cdot a=p^{2\nu j}\cdot a $, i.e., $(p^{\nu j}-1)(p^{\nu j}+1)\cdot a=0$.  Since ${\rm ord}_q(p^\nu)=2j$,  we have $q\nmid (p^{\nu j}-1)$, and hence, $(p^{\nu j}+1)\cdot a=0$.    This implies that $-{a}=p^{\nu j}\cdot{a} \in S_{p^\nu}(a)$.  Therefore,  $S_{p^\nu}(a)$ is of type $\I$ as desired.
\end{proof}

Next, we show that all $p^\nu$-cyclotomic classes of elements in  $A\setminus \{0\}$ are of the same type.   

\begin{proposition}\label{Prop:AllTypeI:Zpks}
    Let $a \in A_q$ and $b \in A_{q^i}$, for any given $1 \leq i \leq k$. Then, $S_{p^\nu}(a)$ is of type $\I$ if and only if $S_{p^\nu}(b)$ is of type $\I$.  Equivalently, $S_{p^\nu}(a)$ is of type $\II$ if and only if $S_{p^\nu}(b)$ is of type $\II$.
\end{proposition}

\begin{proof} Assume that $S_{p^\nu}(a)$ is  of type $\I$.  By Proposition~\ref{Prop:TypeI:Zpks},  it follows that   $ {\rm ord}_q(p^\nu)$ is even.  
By Lemma~\ref{ord},    ${\rm ord}_{q^i}(p^\nu) = {\rm ord}_{q}(p^\nu) q^j$   for some positive integer $j$. Then $q^i|(p^{\nu {\rm ord}_{q}(p^\nu) q^j}-1) $ and   $q^i\nmid (p^{\nu \frac{{\rm ord}_{q}(p^\nu) }{2}q^j}-1) $.  If $q| (p^{\nu \frac{{\rm ord}_{q}(p^\nu) }{2}q^j}-1) $, then  ${\rm ord}_{q}(p^\nu)| \frac{{\rm ord}_{q}(p^\nu) }{2}q^j$ which is impossible since $q$ is odd. Hence, we have $q^i| (p^{\nu \frac{{\rm ord}_{q}(p^\nu) }{2}q^j}+1) $.  Consequently,  $-b=p^{\nu \frac{{\rm ord}_{q}(p^\nu) }{2}q^j} \cdot b$ which means that $S_{p^\nu}(b)$ is of type $I$.

%
%
%
%

    Conversely, assume that $S_{p^\nu}(b)$ is of type $\I$.    
    By \cite[Remark 3.1]{JL2015},  ${\rm ord}_{q^i}(p^\nu) $ is even.
    By Lemma~\ref{ord},  it follows that ${\rm ord}_{q^i}(p^\nu) = {\rm ord}_{q}(p^\nu) q^j$  for some positive integer $j$. Hence, ${\rm ord}_{q}(p^\nu) $ is even.  By Proposition~\ref{Prop:TypeI:Zpks},
      we have   that  $S_{p^\nu}(a)$  is  of type $I$.
\end{proof}

Combining Propositions~\ref{Prop:TypeI:Zpks} and \ref{Prop:AllTypeI:Zpks}, the next corollary  follows immediately.

\begin{corollary}\label{Cor:AllTypeI:Zpks}
    Let $a $ be a non-zero element in $A \cong (\mathbb{Z}_{q^k})^s$. Then $S_{p^\nu}(a)$ is of type $\I$ if and only if $ {\rm ord}_q(p^\nu)$  is even.  Equivalently,   $S_{p^\nu}(a)$ is of type $\II$ if and only if  $ {\rm ord}_q(p^\nu)$  is odd.
\end{corollary}

If   $ {\rm ord}_q(p^\nu)$  is even, then  $S_{p^\nu}(a)$ is of type $\I$ for all $a \in A \setminus \{0\}$. Hence, by  Corollary~\ref{Cor:AllTypeI:Zpks},  it can be deduced that \[r_\I =A_1+ \sum_{i=1}^{k}\frac{\mid{A_{q^i}}\mid}{{\rm ord}_{q^i}(p^\nu)} =1+ \sum_{i=1}^{k}\frac{q^{is}-q^{(i-1)s}}{{\rm ord}_{q^i}(p^\nu)} \text{ and } r_{\II}=0.\]   

 On the other hand, if  $ {\rm ord}_q(p^\nu)$  is odd, then  $S_q(a)$ is of type $\II$ for all $a \in A \setminus \{0\}$. Hence, by Corollary~\ref{Cor:AllTypeI:Zpks}, we have  \[r_\I = |A_1| =1 \text{ and } r_{\II} = \sum_{i=1}^{k}\frac{\mid{A_{q^i}}\mid}{2 {\rm ord}_{q^i}(p^\nu)} = \sum_{i=1}^{k}\frac{q^{is}-q^{(i-1)s}}{2{\rm ord}_{q^i}(p^\nu)}.\]   

In both cases, the number of Euclidean complementary dual abelian codes in $\mathbb{F}_{p^\nu}[( \mathbb{Z}_{q^k})^s]$ follows immediately from Corollary \ref{corNumLCDE}.

\section{Hermitian Complementary Dual Abelian Codes}

Note that  the Hermitian inner product is defined over a finite field of square order.  To study Hermitian  complementary abelian codes, it suffices to  focus on  a group algebra $\mathbb{F}_{{p^{2\nu}}}[G]$, where $p$ is a prime,  $\nu$ is a positive  integer, and  $G$ is a finite abelian group.

\subsection{Characterization}

For a finite abelian group $G$, write $G=A\times P$, where $p\nmid |A|$ and $P$ is a $p$-group. In this section,   the characterization and enumeration of  Hermitian complementary dual abelian codes in   $\mathbb{F}_{p^{2\nu}}[G]$ are given and they are shown to be   independent of $P$.

The decomposition of  $\mathbb{F}_{p^{2\nu}}[A]$ in  \cite{JLS2014} is key to  study  Hermitian complementary dual abelian codes in $\mathbb{F}_{p^{2\nu}}[G]$ and it is recalled as follows. 
For each  $a\in A$, the ${p^{2\nu}}$-cyclotomic class $S_{{p^{2\nu}}}(a)$ is said to be of  {\em type} $\Ip$    if    $S_{{p^{2\nu}}}(a)=S_{p^{\nu}}(-p^\nu  a)$ or {\em type} $\IIp$  if $S_{{p^{2\nu}}}(a)\ne S_{{p^{2\nu}}}(- p^\nu a)$.  
Without loss of generality, the representatives $a_1, a_2, \dots, a_t$ of  $p^{2\nu}$-cyclotomic classes   of $A$  can be  chosen such that $\{a_j| j=1,2,\dots,{r_{\Ip}}\}$ and $\{a_{r_{\Ip}+l}, a_{r_{I}+r_{\IIp}+l}=-p^\nu a_{r_{\Ip}+l} \mid l=1,2,\dots, r_{\IIp}\}$ are  sets of representatives of $p^{2\nu}$-cyclotomic classes of $A$ of types $\Ip$ and ${\IIp}$, respectively, where $t=r_{\Ip}+2r_{\IIp}$.    Assume the notations  used  in Section 2, we have   $s_i=|S_{p^{2\nu}}(a_i)|$ for all $1\leq i\leq t$ and  $s_{r_\Ip+l}=s_{r_\Ip+r_\IIp+l}$ for all $1\leq l\leq r_\IIp$.

Rearranging the terms in the decomposition of  $\mathcal{R}$ in  \eqref{eq-decom01}  based on the    $2$ types of cyclotomic classes  mentioned above (see \cite{JLS2014}),  we have 
\begin{align}
\mathbb{F}_{p^{2\nu}}[A\times P]    & 
\cong   \left( \prod_{j=1}^{r_{\Ip}} \mathbb{K}_j[P]  \right) \times \left( \prod_{l=1}^{r_{\IIp}} (\mathbb{L}_l[P]\times \mathbb{L}_l[P] )  \right), \label{HeqSemiSim}
\end{align}
where   $ \mathbb{K}_j\cong \mathbb{F}_{p^{2\nu s_{j}}}  $ for all   $j=1,2,\dots, r_{\Ip}$  and  $ \mathbb{L}_l  \cong \mathbb{F}_{  p^{2\nu s_{r_\Ip+l}}}   $      for all $l=1,2,\dots, r_{\IIp}$.

From \eqref{HeqSemiSim},  it follows that an abelian code $C$ in  $\mathbb{F}_{p^{2\nu}}[A\times P]   $ can be viewed as 
\begin{align}\label{HdecomC} 
C\cong   \left(\prod_{j=1}^{r_{\Ip}} C_j  \right)\times \left(\prod_{l=1}^{r_{\IIp}} \left( D_l\times D_l^\prime\right) \right), \end{align}
where $C_j$, $D_s$ and $D_s^\prime$ are   abelian    codes in        $\mathbb{K}_j[P]$, $\mathbb{L}_l[P]$ and $\mathbb{L}_l[P]$, respectively,  for all    $j=1,2,\dots,r_{\Ip}$ and  $l=1,2,\dots,r_{\IIp}$.

From   \cite[Section II.D]{JLS2014},  the Hermitian  dual of $C$  in (\ref{HdecomC}) is of the 
form 

\begin{align} \label{Heq-Edual}
C^{\perp_{\rm H}}\cong    \left(\prod_{j=1}^{r_{\Ip}} C_j ^{\perp_{\rm H}} \right)\times \left(\prod_{l=1}^{r_{\IIp}} \left( (D_l^\prime) ^{\perp_{\rm E}}\times  D_l^{\perp_{\rm E}}\right) \right).
\end{align}

The characterization of   Hermitian complementary  dual abelian codes in  $\mathbb{F}_{p^{2\nu}}[G]   $  is given in the following proposition. 

\begin{proposition} \label{HcharLCDE} Let $p$ be a prime and let $\nu$  be a positive integer. Let  $A$ be a finite abelian group such that $p\nmid |A|$ and let $P$ be a finite abelian $p$-group. Then an abelian code $C$ in $\mathbb{F}_{p^{2\nu}}[A\times P]  $ decomposed as in  \eqref{HdecomC} is Hermitian complementary dual if and only if the following statements hold.
    \begin{enumerate}
        \item  $C_j$ is Hermitian  complementary dual   for all $1\leq j\leq r_\Ip$.
        \item  $D_l\cap (D_l^\prime)^{\perp_{\rm E}}=\{0\}$ and    $D_l^\prime \cap D_l^{\perp_{\rm E}}=\{0\}$  for all $1\leq l\leq r_{\IIp}$.
    \end{enumerate}
\end{proposition}
\begin{proof}
    The result can be deduced directly from \eqref{HdecomC} and  \eqref{Heq-Edual}.
\end{proof}

\begin{corollary} \label{HcorCharLCDE}
    Let $p$ be a prime and let $\nu$  be a positive  integer. Let $A$ be a finite abelian group such that $p\nmid |A|$ and let   $P$ be a finite abelian $p$-group. Then an abelian code $C$ in $\mathbb{F}_{p^{2\nu}}[A\times P]  $ decomposed as in  \eqref{HdecomC} is Hermitian complementary dual if and only if the following statements hold.
    \begin{enumerate}
        \item  $C_j\in \{\{0\}, \mathbb{K}_j[P]\}$      for all $1\leq j\leq r_\Ip$.
        \item  $(D_l, D_l^\prime ) \in \{ (\{0\}, \mathbb{L}_j[P]), (\mathbb{L}_j[P],\{0\})\}$  for all $1\leq l\leq r_{\IIp}$.
    \end{enumerate}
\end{corollary}
\begin{proof}
    By Proposition \ref{HcharLCDE},   $C_j$ is Hermitian  complementary dual   for all $1\leq j\leq r_\Ip$ and 
    $D_l$ is a direct summand  for all $1\leq l\leq r_{\IIp}$. By Corollary \ref{Hcomplementary}, we have  $C_j\in \{\{0\}, \mathbb{K}_j[P]\}$      for all $1\leq j\leq r_\Ip$.  By Theorem \ref{derectsummand}, $D_l \in \{ \{0\}, \mathbb{L}_j[P] \}$  for all $1\leq l\leq r_{\IIp}$. Hence, the result follows. 
\end{proof}

The number of Hermitian complementary dual abelian codes in $\mathbb{F}_{p^{2\nu}}[A\times P]  $ is  independent of  $P$ and determined in the following corollary. 
\begin{corollary}\label{HcorNumLCDE}
    Let $p$ be a prime and let $\nu$  be a positive integer. Let $A$ be a finite abelian group such that $p\nmid |A|$ and let   $P$ be a finite abelian $p$-group.  If  $\mathbb{F}_{p^{2\nu}}[A\times P]  $ decomposed as in  \eqref{HeqSemiSim}, then the number of Hermitian complementary dual abelian codes in $\mathbb{F}_{p^{2\nu}}[A\times P]  $ is  \[2^{r_{\Ip}+r_{\IIp}}\]
\end{corollary}
\begin{proof}
    From the characterization in Corollary \ref{HcorCharLCDE},  the number of choices of      $C_j$   is    $2^{r_\Ip}$ and  the number of choices of 
    $(D_l, D_l^\prime )  $  is  $2^{r_{\IIp}}$. Hence, the  number of Hermitian complementary dual abelian codes in $\mathbb{F}_{p^{2\nu}}[A\times P]  $ is $2^{r_{\I}+r_{\II}}$ as desired. 
\end{proof}

\subsection{Enumerations}
As discussed in Corollary \ref{HcorNumLCDE},  the number of Hermitian complementary dual abelian codes in $\mathbb{F}_{p^{2\nu}}[A\times  P]  $ is  independent of $P$. It is therefore sufficient  to  determined the number  ${r_{\Ip}+r_{\IIp}}$ from the semi-simple group algebra  $\mathbb{F}_{p^{2\nu}}[A]$.

Let $\lambda $   be a function defined by 
\begin{align} \label{Hr12}
\lambda(d,p^\nu)=
\begin{cases}
1 &\text{if there exists an odd positive integer } i \text{ such that }  d|(p^{\nu i}+1),\\
0 &\text{otherwise} .\\		
\end{cases}
\end{align}

The  following lemma is a straightforward generalization of   the case where $p=2$  (see {\cite[Lemma 3.5]{JLS2014}}.) 

\begin{lemma}\label{HpropType} Let $p$ be  a prime and let $\nu$ be a positive integer.  Let $A$ be a finite  abelian group    such that $p\nmid |A|$  and let  $a\in A$. Then  $S_{p^\nu}(a) $ is of type $\Ip$  if and only if $\lambda({\rm ord}(a) ,p^\nu)=1$. 
\end{lemma}

\begin{proposition} \label{HpropEnum}
    Let $p$ be a prime and let $\nu$ be a positive integer. Let $A$ be  a finite abelian group of exponent $N$. If $p\nmid N$, then 
    \[{r_{\Ip}+r_{\IIp}}=\sum_{d|N}\lambda(d,p^\nu)\frac{\mathcal{N}_A(d) }{ {\rm ord}_d(p^\nu)} +\frac{1}{2} \sum_{d|N}\left(1-\lambda(d,p^\nu)\right)\frac{\mathcal{N}_A(d) }{ {\rm ord}_d(p^\nu)} ,\]
    where  $\mathcal{N}_A(d)$ denotes the number of elements of order $d$ in $A$ determined in \cite{B1997}.
\end{proposition}
\begin{proof}
    Using the argument similar to those in \cite[Remark 2.5]{JLS2014},  for each $d|N$, the elements of order $d$ in $A$  are partitioned into $p^{2\nu}$-cyclotomic classes of the same size $ \frac{\mathcal{N}_A(d) }{ {\rm ord}_d(p^{2\nu})} $.  By  Lemma \ref{HpropType}, it follows that  \[{r_{\Ip}}=\sum_{d|N}\lambda(d,p^\nu)\frac{\mathcal{N}_A(d) }{ {\rm ord}_d(p^\nu)} \text{  and } r_{\IIp} =\frac{1}{2} \sum_{d|N}\left(1-\lambda(d,p^\nu)\right)\frac{\mathcal{N}_A(d) }{ {\rm ord}_d(p^\nu)} .\]
    Hence, 
    the result follows.
\end{proof}

In the case where $A$ is a cyclic group of order $n$ with $p\nmid n$,    the exponent of $A$ is $n$ and $\mathcal{N}_A(d)=\Phi(d)$ for all divisors $d$ of $n$, where  $\Phi(d)$ denotes the  Euler totient phi function.     The next corollary  follows. 
\begin{corollary}
    Let $p$ be a prime and let $\nu$ be a positive integer. Let $n$ be a positive integer such that  $p\nmid n$. Then  the number of Hermitian complementary dual cyclic codes of length $np^k$ over $\mathbb{F}_{p^{2\nu}}$ equals  the number of Hermitian complementary dual cyclic codes of length $n$ over $\mathbb{F}_{p^{2\nu}}$ for all   integers $k\geq 0$ which  is  equal to
    \[2^{\sum_{d|n}\lambda(d,p^\nu)\frac{\Phi(d) }{ {\rm ord}_d(p^\nu)} +\frac{1}{2} \sum_{d|n}\left(1-\lambda(d,p^\nu)\right)\frac{\Phi(d) }{ {\rm ord}_d(p^\nu)} }.\]
\end{corollary}

From Proposition \ref{HpropEnum},  for an arbitrary finite abelian group $A$ with $p\nmid |A|$,    we have ${r_{\Ip}+r_{\IIp}}\geq 1$.   Hence, there  exist at least two Hermitian complementary dual abelian codes  in  $\mathbb{F}_{p^\nu}[A]$.

 In Table \ref{T1}, the  number ${r_{\Ip}+r_{\IIp}}$ for group algebras $\mathbb{F}_4[A]$ is given for finite abelian  groups $A$ of odd order less than $50$. 
    
    \begin{table}[!hbt]
        \centering
 \begin{tabular}
     {|r|r|r|} \hline Order of $A$ & $A$~~~~~~ & $r_\Ip+r_\IIp$~~\\
     \hline  
     $3$ &$\mathbb{Z}_3$ &$ 3$\\
     $5$ &$\mathbb{Z}_5 $ &$ 2$\\
     $7$ &$\mathbb{Z}_7 $ &$ 2$\\
     $9$ &$\mathbb{Z}_3\times\mathbb{Z}_3 $ &$ 9$\\
     &$\mathbb{Z}_{3^2} $ &$ 5$\\
     $11$ &$\mathbb{Z}_{11}$ &$ 3$\\
     $13$ &$\mathbb{Z}_{13} $ &$2 $\\
     $15$ &$\mathbb{Z}_5\times\mathbb{Z}_3 $ &$ 6$\\
     $17$ &$\mathbb{Z}_{17} $ &$3 $\\
     $19$ &$\mathbb{Z}_{19} $ &$3 $\\
     $21$ &$\mathbb{Z}_7\times\mathbb{Z}_3 $ &$ 6$\\
     $23$ &$\mathbb{Z}_{23} $ &$2$\\
     $25$ &$ \mathbb{Z}_5\times\mathbb{Z}_5 $ &$ 7$\\
     &$\mathbb{Z}_{5^2} $ &$ 3$\\
     $27$ &$\mathbb{Z}_3\times\mathbb{Z}_3\times\mathbb{Z}_3 $ &$27 $\\
     &$\mathbb{Z}_{3^2}\times\mathbb{Z}_3 $ &$ 15$\\
     &$\mathbb{Z}_{3^3}$ &$ 7$\\
     $29$ &$\mathbb{Z}_{29}$ &$ 2$\\
     $31$ &$\mathbb{Z}_{31} $ &$4 $\\
     $33$ &$\mathbb{Z}_{11}\times\mathbb{Z}_3 $ &$ 9$\\
     $35$ &$\mathbb{Z}_{7}\times\mathbb{Z}_5 $ &$ 5$\\
     $37$ &$\mathbb{Z}_{37} $ &$ 2$\\
     $39$ &$\mathbb{Z}_{13}\times\mathbb{Z}_3 $ &$ 6$\\
     $41$ &$\mathbb{Z}_{41} $ &$3 $\\
     $43$ &$\mathbb{Z}_{43}$ &$ 7$\\
     $45$ &$\mathbb{Z}_5\times\mathbb{Z}_3\times\mathbb{Z}_3$ &$18 $\\
     &$\mathbb{Z}_5\times\mathbb{Z}_{3^2} $ &$ 9$\\
     $47$ &$\mathbb{Z}_{47} $ &$2 $\\
     $49$ &$ \mathbb{Z}_7\times\mathbb{Z}_7$ &$ 8$\\
     &$ \mathbb{Z}_{7^2}$&$ 3$\\
     \hline 
 \end{tabular}
        \caption{The number  $r_\Ip+r_\IIp$  for group algebras $\mathbb{F}_4[A]$} \label{T1}
    \end{table}

From the properties of elements in  $A=(\mathbb{Z}_{2^k})^s$ discussed in \cite[Sections 2 and 3]{PJD2018}, a simplified formula  for  \eqref{Hr12} is  given for some families of  finite abelian groups.

\begin{example} Let  $p$ be an odd prime  and let $A=(\mathbb{Z}_{2^k})^s$  for some positive integers $s$ and $k$.  Let $\gamma$ be the  largest integer $0 \leq \gamma \leq k$ such that $2^{\gamma} | (p^\nu + 1)$.  From \cite[Remark 2.7]{PJD2018},  we have   \[r_\Ip = 2^{\gamma s} \text{  and }
    r_{\IIp} =  \sum_{i=\gamma+1}^k\frac{2^{is} - 2^{(i-1)s}}{2{\rm ord}_{2^i}(p^{2\nu}) }.\] 
    Then   $ r_\Ip +  r_{\IIp} = 2^{\gamma s} + \sum_{i=\gamma+1}^k\frac{2^{is} - 2^{(i-1)s}}{2{\rm ord}_{2^i}(p^{2\nu}) }$, and hence, the number of Hermitian complementary dual abelian codes in $\mathbb{F}_{p^{2\nu}}[(\mathbb{Z}_{2^k})^s ]$ is 
    \[2^{2^{\gamma s} + \sum_{i=\gamma+1}^k\frac{2^{is} - 2^{(i-1)s}}{2{\rm ord}_{2^r}(p^{2\nu}) }}.\]
\end{example}

\begin{example} Let $p$ be a prime and let $q$ be an odd prime such that    $\gcd(q,p)=1$. Let  $A=(\mathbb{Z}_{q^k})^s$ for some positive integers $s$ and $k$.  From \cite[Corollary 3.12]{PJD2018}, we have the following two  cases.
    
    \noindent {\bf Case I:} There exists $a \in A \setminus \{0\}$ such that $S_{p^\nu}(h)$ is of type $\Ip$. By \cite[Corollary 3.12]{PJD2018},   we have that that $S_{p^{2\nu}}(a)$ is of type $\Ip$ for all $a\in A$  and  \[r_\Ip = 1+ \sum_{i=1}^{k}\frac{q^{is}-q^{(i-1)s}}{{\rm ord}_{q^i}(p^{2\nu}) } \text{ and } r_{\IIp}=0.\]

    In this case, we have  $r_\Ip+r_\IIp=r_\Ip$, and hence,    the number of Hermitian complementary dual abelian codes in $\mathbb{F}_{p^{2\nu}}[(\mathbb{Z}_{q^k})^s ]$ is 
    \[ 2^{1+ \sum_{i=1}^{k}\frac{q^{is}-q^{(i-1)s}}{{\rm ord}_{q^i}(p^{2\nu}) } }.\]

    \noindent {\bf Case II:}   There exists $a \in A \setminus \{0\}$ such that $S_{p^{2\nu}}(a)$ is of type $\IIp$. By  \cite[Corollary 3.12]{PJD2018} ,  we have  that $S_{p^{2\nu}}(a)$  is of type $\IIp$ for all  $a \in A \setminus \{0\}$,  and hence,  \[r_\Ip   =1 
    \text{ and } r_{\IIp}  = \sum_{i=1}^{k}\frac{q^{is}-q^{(i-1)s}}{2{\rm ord}_{q^i}(p^{2\nu}) }.\]  
    In this case, we have 
    \[r_\Ip+r_\IIp= 1+ \sum_{i=1}^{k}\frac{q^{is}-q^{(i-1)s}}{2{\rm ord}_{q^i}(p^{2\nu}) } ,\] and hence,   the number of Hermitian  complementary dual abelian codes in $\mathbb{F}_{p^{2\nu}}[(\mathbb{Z}_{q^k})^s ]$ is 
    \[ 2^{1+ \sum_{i=1}^{k}\frac{q^{is}-q^{(i-1)s}}{2{\rm ord}_{q^i}(p^{2\nu}) } }.\]
    
\end{example}

\section{Conclusion and Remarks}
In this paper, a family of abelian codes with complementary dual  in group algebras $\mathbb{F}_{p^\nu}[G]$  has been investigated  under both the Euclidean and Hermitian inner products. The characterization of such codes have been  given. Subsequently,   the number of complementary dual abelian codes in  $\mathbb{F}_{p^\nu}[G]$  has been shown to be independent of the Sylow $p$-subgroup of $G$ and it has been completely determined for every  finite abelian group $G$, prime $p$, and positive integer $\nu$.  A simplified formula for the enumeration has been provided in the  case where $G$ is a cyclic group or  a $q$-group with $p\ne q$. 

For application purpose, it is of natural interest to continue the study on the efficiency of codes in this family. Hence, the determination of their minimum distances is an interesting problem as well. 






\begin{thebibliography}{99}
    
    \bibitem{B1997} Benson, S.:  Students ask the darnedest things: A result in elementary group theory. 
    {Math. Mag.}  {\bf 70},  207--211 (1997).
    
    \bibitem{Be1967_2}   Berman, S. D.:
    Semi-simple cyclic and abelian codes. 
    {Kibernetika} {\bf 3},  21--30 (1967).
    
    
    \bibitem{BS2011} Bernal, J. J.,  Sim\'{o}n,   J. J.:  Information sets from defining sets in abelian codes.  {IEEE Trans. Inform. Theory}   \textbf{57},  7990--7999 (2011). 
    
    
    
    \bibitem{CG2015} Carlet, C., Guilley, S.: Complementary dual codes for countermeasures to side-channel attacks.  Coding Theory and Applications {\bf 3},   97--105 (2015).
    
    \bibitem{Carletetal2015} Carlet, C., Daif, A., Danger, J.L., Guilley, S., Najm, Z., Ngo, X.T., Portebouef, T.,  Tavernier, C.: Optimized linear complementary codes implementation for hardware trojan prevention. In: Proceedings of European Conference on Circuit Theory and Design, 2015 August  24-26; Trondheim, Norway. Piscataway, USA: IEEE (2015).
    
    
    \bibitem{Ch1992}  Chabanne, H.:
    Permutation decoding of abelian codes. 
    {IEEE Trans. Inform. Theory}  \textbf{38},  1826--1829 (1992).
    
    
    \bibitem{DKL2000} Ding, C., Kohel, D. R., Ling, S.:
    Split group codes. 
    {IEEE Trans. Inform. Theory}  {\bf 46},    485--495 (2000).
    
    
    \bibitem{EHH2011} Etesami, J., Hu, F.,  Henkel, W.: LCD codes and iterative decoding by projections, a first step towards an intuitive description of iterative decoding. In: Proceedings of IEEE Globecom, 2011 December 5-9 ; Texas, USA. Piscataway, USA: IEEE (2011).
    
    
    
    
    \bibitem{GJG2016} Guenda, K., Jitman, S., Gulliver, T. A.:  Constructions of good entanglement-assisted quantum error correcting codes. {Des., Codes and Cryptogr.}, doi:10.1007/s10623-017-0330-z.
    
    
    \bibitem{ISW2003} Ishai, Y., Sahai, A., Wagner, D.: Private circuits: securing hardware against probing attacks. In: {\em CRYPTO}, vol. 2729 of Lecture Notes in Computer Science, pages 463–-481. Springer, August 17–21 2003. Santa Barbara, CA, USA.
    
      \bibitem{JL2015}  Jitman., S,  Ling, S.: Quasi-abelian codes.  Designs, Codes and Cryptography {\bf  74}, 511--531 (2015).
    
    
    \bibitem{JLLX2012} Jitman, S.,  Ling, S.,  Liu, H.,  Xie, X.:  Abelian codes in principal ideal group algebras.   {IEEE Trans. Inform. Theory} {\bf  59},   3046--3058 (2013). 
    
    
    
    \bibitem{JLS2014}   Jitman, S.,  Ling,  S.,  Sol\'e, P.: Hermitian self-dual Abelian codes. {IEEE Trans. Inform. Theory} {\bf  60}, 1496--1507  (2014).
    
    
    
    
    
    
    
    
    
    \bibitem{M1992} Massey, J.L.: Linear codes with complementary duals. Discrete Mathematics {\bf 106/107} 337--342 (1992).
    
    
    \bibitem{M1997}  Moree, P.: On the divisors of $a^k+b^k$. {Acta Arithmetica} {\bf LXXX}, 197--212 (1997).
    
    \bibitem{Ngoetal2014} Ngo, X.T., Guilley, S., Bhasin, S., Danger, J.L., Najm, Z.: Encoding the state of integrated circuits: a proactive and reactive protection against hardware trojans horses. In: Proceedings of WESS '14, 2014 October 12-17; New Delhi, India.  New York, ACM (2014). 
    
    
    
    \bibitem{Ngoetal2015} Ngo, X.T., Bhasin, S., Danger, J.L., Guilley, S.,  and Najm, Z.:  Linear complementary dual code improvement to strengthen encoded cirucit against Hardware Trojan Horses. In: Proceedings of IEEE International Symposium on Hardware Oriented Security and Trust (HOST): 2015 May 2015; Washington DC Metropolitan Area, USA. Piscataway, USA: IEEE (2015).
    
    
    
    \bibitem{N1972}  Nicholson, W. K.:  {Local group rings}.   {Canad. Math. Bull.}  \textbf{15}, 137--138 (1972).
    
    
    \bibitem{PJD2018}   Palines,  H.  S.,	{Jitman, S.}, Dela Cruz, R. B.:  Hermitian self-dual quasi-abelian codes.  {Journal of Algebra Combinatorics Discrete Structures and Applications}, to appear.
    
    
    \bibitem{RS1992} Rajan, B. S. ,   Siddiqi,  M. U.:
    Transform domain characterization of abelian codes. 
    {IEEE Trans. Inform. Theory}  {\bf  38},  1817--1821 (1992).
    
    
    
    \bibitem{S1993}  Sabin,  R. E.:
    On determining all codes in semi-simple group rings. 
    {Lecture Notes in Comput. Sci.}  \textbf{673},   279--290 (1993).
    
    \bibitem{Sen2004}  Sendrier, N.: Linear codes with complementary duals meet the Gilber-Varshamov bound. Discrete Math {\bf 285},  345--347 (2004).
    
    \bibitem{YM1994} Yang, X., Massey, J.L.: The condition for a cyclic code to have a complementary dual. Discrete Mathematics {\bf 126},  391--393 (1994).  
    
    
    
    
    
    
\end{thebibliography}


\end{document}